\documentclass[10pt]{article}

\usepackage[latin9]{inputenc}
\usepackage{latexsym, bm, amsmath, amssymb, graphics, amsthm, dsfont, enumerate, physics}
\usepackage[a4paper,
            bindingoffset=0.2in,
            left=0.8in,
            right=0.8in,
            top=0.8in,
            bottom=0.8in,
            footskip=.25in]{geometry}
\usepackage{xcolor,placeins}
\usepackage[ruled,vlined]{algorithm2e}
\usepackage{tikz}

\usetikzlibrary{calc,trees,positioning,arrows,chains,shapes.geometric,  decorations.pathreplacing,decorations.pathmorphing,shapes,  matrix,shapes.symbols}
\usepackage{tikz-cd}
\tikzcdset{scale cd/.style={every label/.append style={scale=#1},
    cells={nodes={scale=#1}}}}

\usepackage{amsmath}
\usepackage{amsthm}

\usepackage{hyperref} 
\usepackage{multirow,multicol}

\usetikzlibrary{arrows.meta,
                positioning,
                shapes}
\newtheorem{theorem}{Theorem}

\newtheorem{lemma}{Lemma}[section]

\usepackage{comment}

\usepackage{booktabs}

\usepackage{setspace}

\usepackage{forest}
\usetikzlibrary{arrows.meta}

\usepackage{xcolor}

\usepackage{tcolorbox}

\usepackage{svg}
\usepackage{subcaption}

\usepackage[normalem]{ulem}
\usepackage{xcolor}

\begin{document}

\title{The Quantum Agreement Theorem\footnote{We thank Pierfrancesco La Mura for very important inspiration for this project.  We thank Julio de Vicente, Aram Harrow, and David P\'erez Garc\'ia for useful comments during a MathQI seminar at UCM Madrid, and Mohammad Mehboudi for valuable discussions.  We acknowledge the stimulating environment provided by the W4Q: Quantum Science Conference, where this work was presented, as well as the Quantum Information Journal Club at IFT Madrid. We are grateful to two referees for very important feedback. M.G.D.~and G.S.~acknowledge financial support from the following grants: PID2020-113523GB-I00 and PID2023-146758NB-I00 funded by MICIU/AEI/10.13039/501100011033.  A.B.~acknowledges financial support from NYU Stern School of Business, NYU Shanghai, and J.P. Valles.  ChatGPT, Grok, and MS Copilot were used to help in searching for examples, identifying proof tactics, and editing text.}}

\author
{Mar\'ia Garc\'ia D\'iaz \footnote{Quantum Mathematics, Algorithms and Machine Intelligence (QUMAMI), Departamento de Matem\'atica Aplicada a la Ingenier\'ia Industrial, Universidad Polit\'ecnica de Madrid, 28006 Madrid, Spain, maria.garcia.diaz@upm.es}
\and{Adam Brandenburger \footnote{Stern School of Business, Tandon School of Engineering, NYU Shanghai, New York University, New York, NY 10012, U.S.A., adam.brandenburger@stern.nyu.edu}}
\and{Giannicola Scarpa \footnote{Quantum Mathematics, Algorithms and Machine Intelligence (QUMAMI), Departamento de Sistemas Inform\'aticos, Universidad Polit\'ecnica de Madrid, 28031 Madrid, Spain, g.scarpa@upm.es}}
    }
\date{\today}\maketitle

\thispagestyle{empty}

\begin{abstract}
Intersubjective consistency asks when different observers of a physical system, reasoning from a shared theory, will agree on their probability assessments.  Classical probability theory provides a benchmark in the form of the Agreement Theorem (Aumann, 1976).  We study the quantum mechanics (QM) analog to the classical Agreement Theorem in a finite-dimensional tripartite setting in which two observers obtain information via local projective measurements.  We represent their mutual awareness by common certainty, which is formally an infinite hierarchy of certainty operators concerning the probability of an event of interest.  We first prove that if the observers' measurements commute with one another and with the property of interest, common certainty forces their probability assessments to coincide, thereby recovering a QM analog to the classical Agreement Theorem.  We then construct a one-parameter family of qutrit-qubit-qubit models in which one observer's measurement does not commute with the property and common certainty of disagreement occurs -- a distinctively QM phenomenon.  The mechanism is coherence between branches distinguished by one but not the other observer's measurement.  When measurement outcomes are stored in a classical register, the resulting dephasing removes this coherence and restores agreement.  Finally, we prove that QM does not permit the extreme $0\operatorname{-}1$ configuration in which Alice is certain of a property and is also certain that Bob is certain of its negation.  These results aim to turn debate in QM about observer-dependent facts into precise mathematical conditions for when intersubjective agreement or disagreement in QM can be sustained.
\end{abstract}

\section{Introduction}
The ability to reach a shared view of a physical system is a prerequisite for scientific objectivity.  However, reaching such a view is not automatic.  Different observers of the same system, using the same physical theory, may legitimately assign different probabilities to a common property of interest if they possess different information.  This is a straightforward point, but it naturally suggests two questions: What does it mean for observers to become mutually aware of one another's probability assessments?  And if they do become aware in a suitably defined sense, will they reach a shared view?  The longstanding tenet of ``intersubjective consistency" across observers (Popper, 1959; d'Espagnat, 2011) constitutes a positive answer to the second question at a general philosophical level.  But a mathematical scrutiny, using the tools of formal epistemic analysis, of the conditions for agreement is also important.

This issue of intersubjective consistency takes on particular significance in quantum mechanics (QM), where the role of the observer has been much debated from early on.  Bohr argued that the outcome of an experiment, and even which physical quantities can meaningfully be said to possess values, naturally depend on the experimental context, while Einstein famously objected that such observer-created facts conflict with scientific realism (Held, 1998).  Wigner (1961) sharpened the issue with his Friend paradox.  An observer who has registered a definite outcome may be described by another observer as part of a larger quantum system in superposition -- a scenario suggesting that observer dependence might yield not just intersubjective disagreement but even outright inconsistency.

Elaborating on Wigner's Friend, Frauchiger and Renner (2018) have presented a scenario exhibiting an especially strong form of intersubjective disagreement, raising questions about the consistency of QM.  (See also Brukner, 2018; Nurgalieva and del Rio, 2019; Proietti et al., 2019).  Relational quantum mechanics (Rovelli, 1996; Di Biagio and Rovelli, 2021) holds that physical states are always defined relative to another system -- usually, an observer.  The question of intersubjective agreement has been studied under this approach and a specific postulate, called ``cross-perspective links"  (Adlam and Rovelli, 2023), put forward that ensures agreement about past quantum events.  In quantum Bayesianism (QBism), full intersubjectivity is allowed in any measurement setting (Fuchs and Schack, 2013; Fuchs, Mermin, and Schack, 2014).  QBism, too, is interested in how observers might actively create conditions for agreement among them (Schack, 2024), such as when Alice first performs a measurement and can then predict with probability one the outcome that Bob will observe if he performs a measurement on her post-measurement state.

Our starting point for an examination of agreement and disagreement among observers of a physical system is the classical Agreement Theorem (Aumann, 1976), which states that if two agents share a common prior probability distribution and it is common certainty what probabilities they assign to an event, then those probabilities must agree.  (In what follows, we use the terms ``observer" and ``agent" interchangeably.)  Here, the modality of common certainty of an event means that both agents are certain of the event, both are certain that the other is certain, and so on indefinitely.  The theorem captures a natural intuition about intersubjectivity: If agents' reasoning is fully transparent, then disagreement is not possible.  The classical Agreement Theorem has been used, for example, to analyze Nash equilibrium (Aumann and Brandenburger, 1995; Pacuit and Roy, 2025) and market efficiency (Milgrom and Stokey, 1982; Sebenius and Geanakoplos, 1983; Gizatulina and Hellman, 2019).  However, since it operates under the rules of classical probability theory, the theorem does not automatically transfer to QM, where a new question arises: \textit{Can two observers of a quantum system have common certainty that they assign different probabilities to the same physical property?}

We formulate and prove several theorems that address this question.  Concretely, we consider a tripartite setting where Alice and Bob perform local measurements on their respective subsystems (``laboratories") and there is a third inaccessible subsystem.  In the commuting case we consider first, the projector of the property of interest acts on the third subsystem.  In the noncommuting case we consider later, the property projector acts jointly on the third subsystem and Alice's subsystem.  Alice's and Bob's measurements continue to commute with each other, so the distinction between the two settings is whether both of their measurements commute with the property projector.  Each agent updates their own probability for the property of interest using the Born-L\"uders rule.  We then formalize a quantum version of the common certainty modality.  This is the epistemic condition we use to examine intersubjectivity in the quantum context.  Our first two results establish a sharp distinction between commuting and noncommuting regimes.

\begin{quote}
\textbf{Commutative Agreement Theorem:} In the commuting regime where Alice's and Bob's measurements commute with one another and with the projector of the property of interest, if Alice and Bob have common certainty of their respective probability estimates of that property, then their estimates must agree.
\end{quote}

This result provides a direct quantum analog of the classical Agreement Theorem and shows that classical intersubjective consistency emerges naturally in QM whenever the relevant operators commute.

\begin{quote}
\textbf{Non-Commutative Disagreement Theorem:} In a noncommuting regime where an agent's measurement does not commute with the projector of the property of interest, it is possible that Alice and Bob have common certainty of their respective probability estimates, yet those estimates differ.
\end{quote}

To establish this result, we construct a one-parameter family of states on a shared qutrit-qubit-qubit system and show that coherence between two branches can yield common certainty of disagreement (CCD).  While this result demonstrates the distinctively quantum phenomenon of CCD, our third result establishes that QM does not permit arbitrary intersubjective conflict. 

\begin{quote}
\textbf{$0\operatorname{-}1$ Impossibility Theorem:} There is no positive-probability state in a (finite-dimensional) quantum system where Alice is certain of a property and is also certain that Bob is certain of the negation of that property.
\end{quote}

Unlike the preceding agreement and disagreement results, this impossibility result does not make any commutativity assumptions.  The result validates that our analysis has physical bite -- there are epistemic states of the observers which are ruled out in QM.  It also distinguishes our approach from Frauchiger and Renner (2018), who, in a challenge to the internal consistency of QM, produce the ``paradoxical" $0\operatorname{-}1$ scenario by chaining together statements of Alice's certainty across mutually incompatible measurement branches.  Our result shows that this scenario cannot arise at a single state.  However, we do produce a superquantum no-signaling box (Popescu and Rohrlich, 1994) that exhibits $0\operatorname{-}1$ disagreement.  (By Abramsky and Brandenburger, 2011, this example can also be cast in quasi-classical terms using signed probabilities, as we show.)  So, our investigation arrives at the familiar conclusion that QM occupies a nontrivial position between classical and superquantum theories.  In particular, the phenomenon of extreme disagreement across observers can be added to the well-known list of implausible superquantum behaviors, such as trivialized communication complexity (van Dam, 2005), violations of information causality (Pawlowski et al., 2009), and signaling in the classical limit (Rohrlich, 2014). 

Beyond these main results, we show how classical agreement can be recovered even in a noncommuting regime through the physical act of recording measurement outcomes in a classical register.  This offers a new perspective on the classical Agreement Theorem, which can be understood as emergent from quantum theory via classical recording.  We also examine the robustness of the Commutative Agreement Theorem when its ideal assumptions are relaxed.  First, if Alice and Bob start with different initial states, as might happen because of noise, channel disturbance, or the presence of an eavesdropper, the difference between their probability estimates can be quantitatively bounded (cf.~Hellman, 2013).  Second, if certainty is relaxed to approximate certainty, reflecting finite statistics, detector inefficiencies, or coarse-graining, then disagreement is again bounded (cf.~Monderer and Samet, 1989).  These results show that agreement in the commuting regime is stable under realistic perturbations.

The picture that emerges from our findings is neither that quantum observers must always agree nor that observer-dependent facts are unconstrained.  Common certainty under commutativity forces agreement.  Common certainty under noncommutativity permits disagreement.  Extreme ($0\operatorname{-}1$) disagreement is ruled out.  Classical recording restores agreement.  These results do not select among interpretations of QM, but, by showing, in particular, that CCD arises as a natural consequence of noncommutativity, they provide support for interpretations that treat observer-dependence of facts as central (Rovelli, 1996; Fuchs and Schack, 2013).  In employing the language of formal epistemics, our work may suggest new tools for further development of these interpretations.

We next present our framework and establish our three main results.  We then examine restoration of agreement via a classical register and bounds on disagreement.  We conclude with a discussion of related work and possible future directions.

\section*{Commutative Agreement Theorem}
Consider a setting where two agents, Alice and Bob, have laboratories labeled $A$ and $B$, respectively.  There is a quantum system made up of one subsystem in each of their labs and a third subsystem in a third lab, labeled $C$, which is not accessible to the agents.  Alice and Bob each perform a measurement in their lab, and they are interested in estimating a property of the subsystem in the third lab.  Figure \ref{fig:3part} depicts this setting. 

\begin{figure}[h] 
\centering
\includegraphics[scale=0.6]{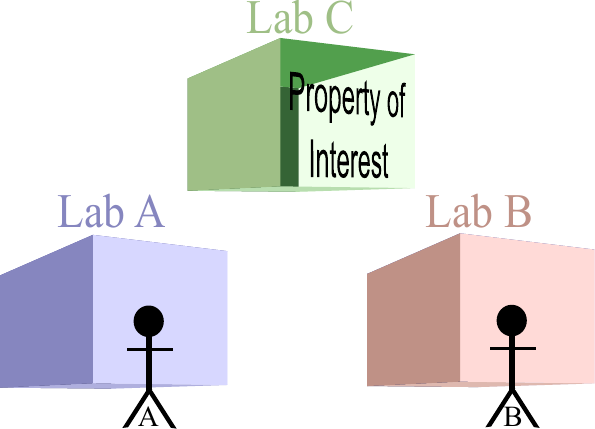}
\caption{Three-laboratory scenario}
\label{fig:3part}
\end{figure}

Formally, we have a tripartite finite-dimensional Hilbert space  ${\mathcal H} = {\mathcal H}_A \otimes {\mathcal H}_B \otimes {\mathcal H}_C$, where the first part is accessible to Alice and the second part to Bob, and the third part is inaccessible.  The quantum state is described by a density operator $\rho \in {\mathcal B}(\mathcal H)$, and this state is prior knowledge shared by Alice and Bob.

Alice and Bob perform projective measurements on their respective systems, in order to estimate the property of interest in the third system ${\mathcal H}_C$.  That is, for each agent $X$ (Alice or Bob), there is a projective measurement $\Pi_X = \{P_X^k : k \in I_X\}$ with projectors $P_X^k\in {\mathcal B}(\mathcal H)$.  Each projector for Alice takes the form $P_A^i = \hat{P}_A^i \otimes \mathbb{I}_B \otimes \mathbb{I}_C$, for $i\in  I_A$, and each projector for Bob takes the form $P_B^j = \mathbb{I}_A \otimes \hat{P}_B^j \otimes \mathbb{I}_C$, for $j \in I_B$, where each $\hat{P}^k_X \in {\mathcal B}(\mathcal H_X)$, and $\mathbb{I}_X$ and $\mathbb{I}_C$ are the identity operators.  (We comment later on extension to a POVM framework.)  The property of interest $E$ in the inaccessible lab is modeled as a projector $P_E = \mathbb{I}_A \otimes \mathbb{I}_B \otimes \hat{P}_E \in {\mathcal B(H)}$ with $\hat{P}_E \in {\mathcal B}(\mathcal H_C)$.  In this setting, each pair of projectors for Alice and Bob commute with each other:
\begin{equation} \label{eq:com1}
[P_A^i, P_B^j] =0 \,\, \text{for all} \,\, i, j,
\end{equation}
and each projector $P_X^k$ commutes with $P_E$:
\begin{equation} \label{eq:com2}
[P_X^k, P_E] = 0 \,\, \text{for all} \,\, k.
\end{equation}

Next, fix a projector $P_F \in {\mathcal B}(\mathcal H)$.  The (conditional) probability that agent $X$ assigns to property $F$, given the outcome $P_X^k$ of their measurement, is given by the Born rule with a L\"uders update:
\begin{equation} \label{eq:born}
\Pr[P_F; P_X^k] = \operatorname{Tr}\Bigl(P_F \frac{P_X^k \, \rho \, P_X^k}{\operatorname{Tr}(P_X^k \, \rho)}\Bigr),
\end{equation}
whenever the denominator is nonzero.  If this probability is $1$, we say that agent $X$ is \emph{certain} of the property $F$ (Monderer and Samet, 1989).  Now let $q$ be a given probability and consider the set of indices $k$ for which agent $X$ assigns probability $q$ to the property $F$, after observing the outcome $P_X^k$:
\begin{equation} \label{eq:index}
K_X(F; q) = \bigl\{k : \operatorname{Tr}(P_X^k \, \rho) > 0 \,\, \text{and} \,\, \Pr[P_F; P_X^k] = q\bigr\}.
\end{equation}
(We suppress the dependence on $\rho$.)  Next define the operator that agent $X$ assigns probability $q$:
\begin{equation} \label{eq:proj_sum}
Q_X(F; q) = \sum_{k\in K_X(F; \, q)} P_X^k,
\end{equation}
which, as a sum of mutually orthogonal projectors, is again a projector.  The probability that agent $Y$ assigns to $Q_X(F; q)$, given their outcome $P_Y^k$, is $\Pr[Q_X(F; q); P_Y^k]$.  If this probability is $1$, we say that agent $Y$ is certain that agent $X$ assigns probability $q$ to property $F$.  A particular instance is $q = 1$, in which case we say agent $Y$ is certain that agent $X$ is certain of property $F$.  We also write $C_X(F)$ for $Q_X(F; 1)$.

We are now ready to define common certainty.  Fix probabilities $q_A$ and $q_B$ for Alice and Bob, respectively.  Let:
\begin{align}
A_0 &= Q_A(E; q_A), \label{eq:init-a} \\
B_0 &= Q_B(E; q_B), \label{eq:init-b}
\end{align}
and recursively define higher-order certainty operators, for all $n \ge 0$:
\begin{align}
A_{n+1} &= A_n \, C_A(B_n), \label{eq:recur-a} \\
B_{n+1} &= B_n \, C_B(A_n). \label{eq:recur-b}
\end{align}
(See the Supplementary Material for details.)  By finiteness of the projective measurements $\Pi_X$, for $X = A, B$, there is a finite integer $N$ such that the recursion stabilizes: $A_n = A_N$ and $B_n = B_N$ for all $n \ge N$.  Writing $A_* = A_N$ and $B_* = B_N$, we call $C_* = A_* \, B_*$ the \emph{common certainty} projector.  It is the fixed point of the certainty recursion, capturing the infinite iteration ``Alice is certain that Bob is certain that Alice is certain $\dots$ that the property $E$ has probability $q_A$ (resp.~$q_B$)."  Since $A_*$ and $B_*$ commute,  we can also write $C_* = B_* \, A_*$, a fact we will use later.  If $q_A = q_B = 1$, we simply say that there is common certainty of $E$.  We now have the ingredients to define the epistemic condition of interest.
\vspace{0.1in}

\noindent \textbf{Definition:} \textit{Fix a projector $P_E$ of a property of interest and probability estimates $q_A$ and $q_B$ for Alice and Bob, respectively.  Then $q_A$ and $q_B$ are common certainty at state $\rho$ if $\operatorname{Tr}(C_* \, \rho) > 0$.} 
\vspace{0.1in}

This definition ensures the common-certainty recursion is non-vacuous by requiring that the operator $C_*$ has positive weight at the given state $\rho$.
\vspace{0.1in}

\noindent \textbf{Commutative Agreement Theorem:} \textit{If it is common certainty at state $\rho$ that Alice assigns $q_A$ to $E$ and Bob assigns $q_B$ to $E$, then $q_A = q_B$.}
\vspace{0.1in}

The proof has four main steps.  In the Supplementary Material, we first show that common certainty implies that Bob's stabilized certainty operator $B_*$ does not disturb the normalized state obtained by conditioning on Alice's stabilized certainty projector $A_*$.  Formally:
\begin{equation} \label{eq:nondisturb}
\frac{B_* A_* \, \rho \, A_* B_*}{\operatorname{Tr}(B_* A_* \, \rho)} = \frac{A_* \, \rho \, A_*}{\operatorname{Tr}(A_* \, \rho)}.
\end{equation}
Second, this invariance ensures that the probability of the property $E$ is the same whether evaluated before or after Bob's certainty operator is applied:
\begin{equation} \label{eq:qA_expression}
\operatorname{Tr}\Bigl(P_E \frac{A_* \, \rho \, A_*}{\operatorname{Tr}(A_* \, \rho)}\Bigr) = \operatorname{Tr}\Bigl(P_E \frac{B_* A_* \, \rho \, A_* B_*}{\operatorname{Tr}(B_* A_* \, \rho)}\Bigr).
\end{equation}
Third, by symmetry, the same reasoning holds with Alice and Bob interchanged:
\begin{equation} \label{eq:qB_expression}
\operatorname{Tr}\Bigl(P_E \frac{B_* \, \rho \, B_*}{\operatorname{Tr}(B_* \, \rho)}\Bigr) = \operatorname{Tr}\Bigl(P_E \frac{A_* B_* \, \rho \, B_* A_*}{\operatorname{Tr}(A_* B_* \, \rho)}\Bigr).
\end{equation}
Since $A_*$ and $B_*$ commute, the right-hand sides of Equations \ref{eq:qA_expression} and \ref{eq:qB_expression} are equal.  It follows that the left-hand sides are also equal.

The final step is to observe that because $P_E$ commutes with Alice's measurement projectors, for every positive-probability outcome projector $P_A^i$ occurring in $A_*$, Alice's assessment satisfies:
\begin{equation} \label{eq:final}
q_A = \operatorname{Pr}[P_E; P_A^i] = \operatorname{Tr}\Bigl(P_E \frac{A_* \, \rho \, A_*}{\operatorname{Tr}(A_* \, \rho)}\Bigr).
\end{equation}
The analogous equation holds for every positive-probability outcome projector for Bob occurring in $B_*$.  Combining these equalities with the equality of the left-hand sides of Equations \ref{eq:qA_expression} and \ref{eq:qB_expression}, we conclude that $q_A = q_B$.

An important point is that even though the recursion of certainty operators stabilizes at a finite stage, owing to the finiteness of the measurement sets, the full force of common certainty is essential.  It is already known in the classical setting that, for any finite order $n$, there are scenarios where Alice and Bob achieve $n$th-order certainty of their respective probability estimates, while still disagreeing on their values (Geanakoplos and Polemarchakis, 1982; Aumann and Brandenburger, 1995).  A general statement of our theorem therefore requires the complete, infinite hierarchy of certainty -- captured by the common-certainty operator $C_*$ -- to ensure agreement.  We also note that common-certainty probabilities are not, in general, the same as pooled probabilities that are obtained if Alice and Bob combine their respective measurement outcomes.  (See the Supplementary Material.)

The three-laboratory arrangement of Figure \ref{fig:3part} provides a natural physical scenario in which Equations \ref{eq:com1}--\ref{eq:com2} hold.  However, the proof uses only these commutativity conditions and holds regardless of which subsystems $E$ concerns.  The Supplementary Material describes additional scenarios in which the same conditions apply.

We now provide an example to show that the certainty recursion does not necessarily stabilize immediately.  This establishes that our commutative theorem has nontrivial epistemic content in that agents may need to refine what they are certain about.
\vspace{0.1in}

\noindent \textbf{Example 1:} \textit{Alice and Bob each perform local rank-$2$ projective measurements on their respective subsystems.  Specifically, Alice's Hilbert space decomposes into two orthogonal two-dimensional subspaces and Bob's Hilbert space into three orthogonal two-dimensional subspaces:}
\begin{equation}
{\mathcal H}_A = {\mathcal A}_0 \oplus {\mathcal A}_1, \,\,\,\, {\mathcal H}_B = {\mathcal B}_0 \oplus {\mathcal B}_1 \oplus {\mathcal B}_2.
\end{equation}
\textit{For each $i, j$, let $\ket{\Phi_{ij}}$ denote a maximally entangled Bell state on the two-dimensional subspace $\mathcal A_i \otimes \mathcal B_j$.  The global state involves the inaccessible qubit in $\mathcal H_C$, which is entangled with the sectors of $\mathcal H_A$ and $\mathcal H_B$.  Specifically, set $\rho = \ketbra{\psi}$ where:}
\begin{equation}
\ket\psi = \frac{1}{\sqrt{3}} \left( \ket{\Phi_{00}} \ket{0}_{C} + \ket{\Phi_{11}} \ket0_{C} + \ket{\Phi_{12}} \ket1_{C} \right). 
\end{equation}

\textit{Alice measures $\Pi_A = \{P_A^0, P_A^1\}$ with $P_A^i = \Pi_{A_i} \otimes \mathbb{I}_B \otimes \mathbb{I}_C$, while Bob measures $\Pi_B = \{P_B^0, P_B^1, P_B^2\}$ with $P_B^j = \mathbb{I}_A \otimes \Pi_{B_j} \otimes \mathbb{I}_C$, where $\Pi_{A_i}$ and $\Pi_{B_j}$ are the orthogonal projectors onto the respective two-dimensional subspaces $\mathcal A_i$ and $\mathcal B_j$.  The property of interest concerns the qubit in lab $C$:}
\begin{equation}
P_E(\theta) = \mathbb{I}_A \otimes \mathbb{I}_B \otimes \ketbra{\phi_\theta} \,\, \text{where} \,\, \ket{\phi_\theta} = \cos\Bigl(\frac{\theta}{2}\Bigr) \ket0 + \sin\Bigl(\frac{\theta}{2}\Bigr) \ket1.
\end{equation}
\textit{The measurement projectors for Alice and Bob commute with each other and with $P_E(\theta)$.  Calculation for the positive-probability outcomes yields:}
\begin{align}
\operatorname{Pr}[P_E(\theta); P_A^0] &= \operatorname{Pr}[P_E(\theta); P_B^0] = \operatorname{Pr}[P_E; P_B^1] = \cos^2\Bigl(\frac{\theta}{2}\Bigr), \\
\operatorname{Pr}[P_E(\theta); P_A^1] &= \frac{1}{2}, \,\, \operatorname{Pr}[P_E(\theta); P_B^2] = \sin^2\Bigl(\frac{\theta}{2}\Bigr).
\end{align}
\textit{Therefore, choosing $q_A = q_B = \cos^2 (\theta/2)$ where $\theta \not= \pi/2 + \pi z$, for an integer $z$, we have $A_0 = P_A^0$ and $B_0 = P_B^0 + P_B^1$.  A further calculation gives $C_A(B_0) = P_A^0$ and $C_B(A_0) = P_B^0$, from which, using Equations \ref{eq:recur-a}--\ref{eq:recur-b}, we find $A_1 = P_A^0$ and $B_1 = P_B^0$.  The recursion stabilizes at this step, so that $C_* = A_* B_* = P_A^0 P_B^0$.  Moreover, $\operatorname{Tr}(C_* \rho) = 1/3\ > 0$, so the probability assessment $\cos^2 (\theta/2)$ is common certainty at $\rho$.}
\vspace{0.1in}

We see in this example that the recursion does not stabilize immediately.  (Variants of the example could be designed that require additional steps until stabilization.)  Also, we obtain a range of examples between $q = 1$ ($\theta = 0$) and $q = 0$ ($\theta = \pi$), showing that our Commutative Agreement Theorem is not restricted to boundary cases.

\section*{Non-Commutative Disagreement Theorem}
We keep the three-laboratory tensor-product structure and Alice's and Bob's local measurements, so that the commutativity condition of Equation \ref{eq:com1} continues to hold.  Now, however, we allow the property of interest to act jointly on the inaccessible system $C$ and Alice's subsystem.  Equation \ref{eq:com2} for Alice no longer holds.  This breaks the Commutative Agreement Theorem and allows the distinctive QM phenomenon of common certainty of disagreement (CCD) to arise.
\vspace{0.1in}

\noindent \textbf{Non-Commutative Disagreement Theorem:} \textit{In the noncommuting regime, it can be common certainty at a state $\rho$ that Alice assigns $q_A$ to $E$ and Bob assigns $q_B$ to $E$, while $q_A \not= q_B$.  In particular, there is a tripartite system, fixed measurements for Alice and Bob, a fixed property $E$, and a family of states $\rho_t$, for $0 \le t \le 1$, such that it is common certainty at $\rho_t$ that $q_A = 1/3$ and $q_B = (1 + t)/3$.}
\vspace{0.1in} 

The theorem is stated to show that CCD can arise across a parametrized family of states and is not a knife-edge phenomenon arising only from a specially constructed state.  The specific parametric family chosen leads to a straightforward proof, which we now outline.  (The Supplementary Material contains details.)  Consider again the three-lab setting.  Alice holds a qutrit, Bob holds a qubit, and the third lab contains an inaccessible qubit:
\begin{equation}
{\mathcal H}_A = {\mathcal A}_0 \oplus {\mathcal A}_1 \oplus {\mathcal A}_2, \,\,\,\, {\mathcal H}_B = {\mathcal B}_0 \oplus {\mathcal B}_1, \,\,\,\, {\mathcal H}_C = {\mathcal C}_0 \oplus {\mathcal C}_1,
\end{equation}
where each ${\mathcal A}_i = \operatorname{span}\{\ket{i}_A\}$, for $i = 0, 1, 2$, each ${\mathcal B}_j = \operatorname{span}\{\ket{\beta_j}_B\}$, for $j = 0, 1$, where $\{\ket{\beta_0}_B, \ket{\beta_1}_B\}$ is an orthonormal basis of ${\mathcal H}_B$, and each ${\mathcal C}_k = \operatorname{span}\{\ket{k}_C\}$, for $k = 0, 1$.  Define the two orthogonal branches:
\begin{equation}
\ket{u} = \ket{0}_A \ket{\beta_0}_B \ket{0}_C \,\, \text{and} \,\, \ket{v} = \ket{1}_A \ket{\beta_0}_B \ket{1}_C,
\end{equation}
which will give Alice different measurement outcomes but Bob the same outcome $\beta_0$.  For $0 \le t \le 1$, define the global state:
\begin{equation}
\rho_t = \frac{1}{2} \ketbra{u}{u} + \frac{1}{2} \ketbra{v}{v} +  \frac{t}{2} \ketbra{u}{v} + \frac{t}{2} \ketbra{v}{u}.
\end{equation}

Alice measures $\Pi_A = \{P_A^0, P_A^1, P_A^2\}$ with $P_A^i = \Pi_{A_i} \otimes \mathbb{I}_B \otimes \mathbb{I}_C$, while Bob measures $\Pi_B = \{P_B^0, P_B^1\}$ with $P_B^j = \mathbb{I}_A \otimes \Pi_{B_j} \otimes \mathbb{I}_C$.  The property of interest is:
\begin{equation}
P_E = \ketbra{\phi}_{AC} \otimes {\mathbb I}_B \,\, \text{where} \,\, \ket{\phi}_{AC} = \frac{1}{\sqrt 3} \bigl( \ket{0}_A \ket{0}_C + \ket{1}_A \ket{1}_C + \ket{2}_A \ket{0}_C \bigr).
\end{equation}
All projectors for Alice and Bob commute, Bob's projectors commute with $P_E$, but Alice's measurement does not commute with $P_E$.  For the positive-probability outcomes, calculation yields:
\begin{equation}
\operatorname{Pr}[P_E; P_A^0] = \operatorname{Pr}[P_E; P_A^1] = \frac{1}{3}, \,\,\,\, \operatorname{Pr}[P_E; P_B^0] = \frac{1 + t}{3}.
\end{equation}
It follows that $A_0 = Q_A(E; 1/3) = P_A^0 + P_A^1$ and $B_0 = Q_B(E; (1 + t)/3) = P_B^0$.  Moreover, we find $C_A(B_0) = P_A^0 + P_A^1 = A_0$ and $C_B(A_0) = P_B^0 = B_0$, so that the common-certainty recursion stabilizes immediately and $C_*=A_0 \, B_0$.  Finally, $\operatorname{Tr}(C_* \, \rho_t) = 1 > 0$, establishing common certainty of $q_A = 1/3$ and $q_B = (1 + t)/3$.

The key to the construction is the interaction between noncommutativity and coherence.  On the support of $\rho_t$, the property projector has the matrix elements:
\begin{equation}
\bra{u} P_E \ket{u} = \bra{v} P_E \ket{v} = \bra{u} P_E \ket{v} = \bra{v} P_E \ket{u} = \frac{1}{3}.  
\end{equation}
The nonzero off-diagonal elements between the two branches imply that $P_E$ does not commute with Alice's measurement, which distinguishes those branches.  By contrast, Bob's measurement assigns both branches the same outcome $\beta_0$.  The state $\rho_t$ contains coherence $t/2$ in each direction between the two branches.  Conditioning on Alice's outcome selects one branch and yields $q_A = 1/3$.  Conditioning on Bob's outcome retains both branches and their coherence, so the two coherence terms raise Bob's probability assignment by $t/3$ relative to Alice.  In sum, noncommutativity makes the coherence contribution possible, and the parameter $t$ determines its magnitude.

An important aspect of our Non-Commutative Disagreement Theorem is that we do not include in the agents' information states any variable recording measurement order or stage of protocol execution, nor any higher-order beliefs about such facts.  Bob's probability assignment is obtained by conditioning the common prior $\rho_t$ only on his own outcome.  By no-signaling, and absent any public protocol information or other classical side information, Bob cannot infer from his local quantum data whether Alice has already measured.  One could imagine a richer model in which Bob assigns probabilities to such protocol facts, to Alice's beliefs about them, to Bob's beliefs about Alice's beliefs, and so on.  The most natural interpretation of our theorem is that it is ``protocol-agnostic."  In an elaborated model with explicit order uncertainty, one may view it as the special case in which it is common certainty that Bob assesses $P_E$ before Alice's measurement is incorporated.  In the Discussion section, we consider the possibility of extending the protocol-agnostic framework to an epistemic model in which measurement order is an additional object of uncertainty.

\section*{Common Certainty with a Classical Register}
The boundary case $t = 0$ of our construction in the previous section isolates the role of coherence.  Although $P_E$ still does not commute with Alice's measurement, the state $\rho_0$ contains no coherence between $\ket{u}$ and $\ket{v}$, and both agents assign probability $1/3$.  Formally, writing dephasing in Alice's measurement as:
\begin{equation}
\mathcal D_A(\rho_t) = \sum_{i = 0}^2 P_A^i \, \rho_t \, P_A^i,
\end{equation}
we get:
\begin{equation}
\mathcal D_A(\rho_t) = \frac{1}{2} \ketbra{u}{u} + \frac{1}{2} \ketbra{v}{v} = \rho_0.
\end{equation}

Recording Alice's outcome in orthogonal states of a classical register achieves this dephasing on the original system while retaining the outcome in the enlarged recorded state.  This observation motivates a move from the protocol-agnostic CCD construction to a protocol-explicit setting.  We suppose that the agents' measurement outcomes are recorded in a classical register encoding the transcript of a shared recording protocol.  Once such a register is included, facts about which measurements have been performed and which outcomes have occurred become part of the physical state, and the agents can condition on them through commuting register projectors.  The calculation above is an instance of the general result that if the agents' assessments are common certainty relative to such a classical record, then they must agree.  This result sheds new light on the classical Agreement Theorem, which, in a quantum world, emerges from the physical act of classical recording.

To proceed, let $R$ denote the finite set of classical transcripts from protocol runs.  The classical register corresponds to the Hilbert space:
\begin{equation}
{\mathcal H}_R = \operatorname{span}\{\ket{r}_R : r \in R\},
\end{equation}
with the associated commuting family of register projectors:
\begin{equation} \label{eq:regalg}
{\mathcal A}_R = \bigl\{\mathbb{I}_{\mathcal H} \otimes \sum_{r \in S} \ketbra{r}_R : S \subseteq R \bigr\},
\end{equation}
where $\mathbb{I}_{\mathcal H}$ is the identity on the original system Hilbert space.  The recording process is modeled as a completely positive trace-preserving (CPTP) map $J : {\mathcal B}({\mathcal H}) \to {\mathcal B}({\mathcal H} \otimes {\mathcal H}_R)$, defined by:
\begin{equation}
\rho^\prime = J(\rho) = \sum_{r \in R} M_r \, \rho \, M_r^\dagger \otimes \ketbra{r}_R,
\end{equation}
where $M_r$ is the Kraus operator linked to transcript $r$.  The resulting state $\rho^\prime$ is block-diagonal in the pointer basis of ${\mathcal H}_R$, thereby representing a measurement with
classically recorded outcomes.

We next demonstrate that, once the outcomes are embedded in this commuting family, the certainty recursion satisfies
the commutativity conditions of the Commutative Agreement Theorem.  To make this explicit, write:
\begin{equation}
f_A : R \to \{1, \ldots, \ell_A\} \,\, \text{and} \,\, f_B : R \to \{1,\ldots,\ell_B\},
\end{equation}
for the functions extracting from each transcript $r$ the outcomes registered by Alice and Bob.  Define their respective projectors on the register:
\begin{equation}
\Pi^{(i, A)}_R = \sum_{\{r \, : \, f_A(r) = i\}} \ketbra{r}_R \,\,\, \text{and} \,\,\, \Pi^{(j, B)}_R = \sum_{\{r \, : \, f_B(r) = j\}} \ketbra{r}_R.
\end{equation}
In the full system, these operators lie in the commuting family ${\mathcal A}_R$ of Equation \ref{eq:regalg}, and therefore, for all pairs of classical outcomes $(i, j)$:
\begin{equation}
\bigl[\mathbb{I}_{\mathcal H} \otimes \Pi^{(i,A)}_R, \mathbb{I}_{\mathcal H} \otimes \Pi^{(j,B)}_R \bigr] = \bigl[\mathbb{I}_{\mathcal H} \otimes \Pi^{(i, A)}_R, P_E \otimes \mathbb{I}_R \bigr] = \bigl[\mathbb{I}_{\mathcal H} \otimes \Pi^{(j, B)}_R, P_E \otimes \mathbb{I}_R \bigr] = 0.
\end{equation}
This says that, on the enlarged Hilbert space, the lifted projectors satisfy the commutativity assumptions of Equations \ref{eq:com1}--\ref{eq:com2}.

For any projector $P \in \mathcal{A}_R$ and any property projector $P_F$ acting on the system, the conditional probability of $F$ given $P$ is given by the analogous formula to Equation \ref{eq:born}:
\begin{equation}
\operatorname{Pr}[P_F \otimes \mathbb{I}_R; P] = \frac{\operatorname{Tr}\bigl[(P_F \otimes \mathbb{I}_R) \, P \, \rho^\prime \, P \bigr]}{\operatorname{Tr}(P \, \rho^\prime)},
\end{equation}
where we use the recorded state $\rho^\prime$.  With this definition, we can proceed to define analogs of the earlier certainty recursion of Equations \ref{eq:init-a}--\ref{eq:recur-b}.  Letting $C_*^\prime $ be the stabilized common-certainty operator, our Commutative Agreement Theorem carries over directly.
\vspace{0.1in}

\noindent \textbf{Definition:} \textit{Fix a property of interest $E$ and probability estimates $q_A^\prime$ and $q_B^\prime$ for Alice and Bob, respectively.  Then $q_A^\prime$ and $q_B^\prime$ are recorded common certainty at (recorded) state $\rho^\prime$ if $\operatorname{Tr}(C_*^\prime \, \rho^\prime) > 0$.}
\vspace{0.1in}

\noindent \textbf{Classical Register Theorem:} \textit{If it is recorded common certainty at state $\rho^\prime$ that Alice assigns $q_A^\prime$ to $E$ and Bob assigns $q_B^\prime$ to $E$, then $q_A^\prime = q_B^\prime$.}
\vspace{0.1in}

We illustrate the theorem by returning to our one-parameter CCD
family.
\vspace{0.1in}

\noindent \textbf{Example 2:} \textit{We keep the measurements and property projector $P_E$ from the CCD construction and work on the enlarged space $\mathcal H \otimes \mathcal H_R$.  Write each transcript as $r = (i, j)$, for $i = 0, 1, 2$ and $j = 0, 1$.  Since Alice's and Bob's measurements commute, the Kraus operator associated with transcript $(i, j)$ is $M_{(i, j)} = P_A^i \, P_B^j$.  There are two positive-probability transcript branches, with:}
\begin{equation}
M_{(0, 0)} \, \rho_t \, M_{(0,0)}^\dagger = \frac{1}{2} \ketbra{u}{u} \,\, \text{and} \,\, M_{(1, 0)} \, \rho_t \, M_{(1,0)}^\dagger = \frac{1}{2} \ketbra{v}{v}.
\end{equation}
\textit{Each branch has probability $1/2$.  We now have:}
\begin{equation}
\rho_t^\prime = \frac{1}{2} \ketbra{u}{u} \otimes \ketbra{(0, 0)}{(0, 0)}_R + \frac{1}{2} \ketbra{v}{v} \otimes \ketbra{(1, 0)}{(1, 0)}_R.
\end{equation}
\textit{The two coherence terms in $\rho_t$ do not appear in $\rho_t'$.  Therefore, the recorded state is independent of $t$ and $\operatorname{Tr}_R(\rho_t^\prime) = \rho_0$.}

\textit{The two transcript branches remain distinct in the classical
register.  Alice's recorded outcomes $0$ and $1$ select the branches
$(0, 0)$ and $(1, 0)$, with updated states $\ketbra{u}{u}$ and $\ketbra{v}{v}$, respectively, so that:}
\begin{equation}
\Pr\bigl[P_E \otimes \mathbb{I}_R; \mathbb{I}_{\mathcal H} \otimes \Pi_R^{(0, A)}\bigr] = \Pr\bigl[P_E \otimes \mathbb{I}_R; \mathbb{I}_{\mathcal H} \otimes \Pi_R^{(1, A)}\bigr] = \frac{1}{3}.
\end{equation}
\textit{Bob's recorded outcome is $0$ on both transcript branches, and so his outcome projector $\Pi_R^{(0, B)}$ selects the union of the two positive-probability transcript sectors.  After tracing out the register, the corresponding state of the original system is the classical mixture:}
\begin{equation} \label{eq:classical-mixture}
\frac{1}{2} \ketbra{u}{u} + \frac{1}{2} \ketbra{v}{v} = \rho_0,
\end{equation}
\textit{so that:}
\begin{equation}
\Pr\bigl[P_E \otimes \mathbb{I}_R; \mathbb{I}_{\mathcal H} \otimes \Pi_R^{(0, B)}\bigr] = \frac{1}{3}.
\end{equation}

\textit{With the resulting projectors for probability assessments, we find the recorded common-certainty recursion stabilizes immediately at $C_*^\prime$ with $\operatorname{Tr}(C_*^\prime \, \rho) = 1$, and so $q_A^\prime = q_B^\prime = 1/3$ is recorded common certainty at $\rho^\prime$.}
\vspace{0.1in}

\begin{table}[h!]
\centering
\renewcommand{\arraystretch}{1.3}
\begin{tabular}{lc}
\toprule
\textbf{Scenario} & \textbf{Probability of $E$} \\
\midrule
Prior (pre-measurement)
& $\operatorname{Tr}(P_E \, \rho_t) =\frac{1 + t}{3}$ \\
Common certainty of disagreement
& Alice: $q_A = \frac{1}{3}$, Bob: $q_B = \frac{1 + t}{3}$ \\
Recorded assessments (post-register)
& Alice: $q_A^\prime = \frac{1}{3}$, Bob: $q_B^\prime = \frac{1}{3}$ \\
\bottomrule
\end{tabular}
\caption{Noncommuting and Classical-Register Probabilities}
\label{tab:1}
\end{table}

Table \ref{tab:1} summarizes the calculations for Example 2.  Classical recording replaces the CCD assessments $q_A$ and $q_B$ with the recorded assessments $q_A^\prime$ and $q_B^\prime$.  Note that, for $t > 0$, the agreed-on assessment after recording is different from the prior probability of $E$, because recording has removed the coherence terms.

\section*{Bounds on Disagreement}
We now examine the question of the extent of disagreement across agents allowed by QM.  We first state a result that applies to the general case, commuting or noncommuting, that rules out ``extreme" disagreement between agents.  This is true even for second-level certainty, let alone common certainty.  (The theorem generalizes an earlier result in Contreras-Tejada et al., 2021, where the $0\operatorname{-}1$ phenomenon is called singular disagreement.)
\vspace{0.1in}

\noindent \textbf{$\mathbf{0\operatorname{-}1}$ Impossibility Theorem:} \textit{For every state $\rho$, there is no outcome of Alice's measurement with positive probability under $\rho$ at which Alice is certain of $E$ and certain that Bob is certain of not-$E$.  Equivalently:}
\begin{equation} \label{eq:impo}
\operatorname{Tr}\bigl(Q_A(E; \, 1) \, C_A(Q_B(E; \, 0)) \, \rho \bigr) = 0.
\end{equation}

This result does not use any commutativity assumptions and therefore imposes an epistemic constraint independent of any compatibility conditions.  Our impossibility theorem has bite.  The Supplementary Material contains an example of a no-signaling box (Popescu and Rohrlich, 1994) where 0-1 disagreement does arise: Alice is certain of $E$ and certain that Bob is certain of not-$E$.  We also present an equivalent formulation of the example on phase space where the sole non-classicality is the use of signed probabilities (Abramsky and Brandenburger, 2011).

Next, we show that our Commutative Agreement Theorem is stable under small disagreement about the global state.  We first examine the case where Alice and Bob start from different initial states $\sigma$ and $\tau$, in place of the same initial state $\rho$ as previously.  Such a difference might arise from, for example, noise, channel disturbance, or the presence of an eavesdropper.  Building on a classical analysis by Hellman (2013), it can be shown that in our commuting regime:
\begin{equation} \label{eq:bound-a}
|q_A - q_B| \le \frac{2 \, \|\sigma - \tau\|_1}{\operatorname{max}\{\operatorname{Tr}(B_* A_* \, \sigma), \operatorname{Tr}(A_* B_* \tau)\}},
\end{equation}
where $\|\cdot\|_1$ is the trace norm of the difference between two quantum states.  This formula allows us to incorporate perturbations of the quantum state.  For example, in the eavesdropping scenario, after Alice has performed her measurement, a third agent, Eve, disturbs the state from $\sigma$ to $\tau$ before Bob makes his measurement.  Inequality \ref{eq:bound-a} shows that the resulting departure from agreement is controlled.

The second robustness test relaxes the certainty modality.  Given a measurement outcome  $P_X^k$ for agent $X$, we follow the classical treatment in Monderer and Samet (1989) and say that $X$ is \textit{$(1 - \epsilon)$-certain} of the property $F$ if, in place of probability $1$ in Equation \ref{eq:born}, we have:
\begin{equation} \label{eq:born-epsilon}
\Pr[P_F; P_X^k] \ge 1 - \epsilon,
\end{equation}
for a small positive number $\epsilon$.  Using Equation \ref{eq:born-epsilon} in place of Equation \ref{eq:born} in the certainty recursion of Equations \ref{eq:init-a}--\ref{eq:recur-b}, we arrive at the condition of \textit{common $(1 - \epsilon)$-certainty} of $F$.  From this, one can show in our commuting regime: If it is common $(1 - \epsilon)$-certainty at state $\rho$ that Alice assigns $q_A$ to $E$ and Bob assigns $q_B$ to $E$, then:
\begin{equation} \label{eq:bound-b}
|q_A - q_B| \le 2\epsilon.
\end{equation}
This second bound is relevant because, in practice, agents must confront finite statistics, detector inefficiencies, and coarse-grained outcomes, making $(1 - \epsilon)$-certainty more realistic.  Inequality \ref{eq:bound-b} shows that the departure from agreement is again smooth.  Proofs of both robustness results are in the Supplementary Material.

\section*{Discussion}
Our results operate at the intersection of the logic of belief and the physics of measurement.  We establish that the quantum world supports an analog of the Agreement Theorem of classical epistemics (Aumann, 1976) in the commuting case, while we find the distinctive epistemic phenomenon of CCD in the noncommuting regime.  We therefore identify a clear boundary between intersubjective disagreement in the classical setting (ruled out) and in the quantum setting (allowed under noncommutativity).  We show that CCD reduces to agreement under classical recording, offering a way to understand the classical Agreement Theorem in a quantum world.

Importantly, QM bounds the disagreement that is possible across agents, via our $0\operatorname{-}1$ impossibility theorem.  This result validates that our epistemic analysis has physical bite -- there are epistemic states ruled out by QM.  The impossibility result also distinguishes our approach from Frauchiger and Renner (2018).  Their work shows that chaining together statements of Alice's certainty across mutually incompatible measurement branches yields a scenario that looks exactly like $0\operatorname{-}1$ disagreement, which is then argued to suggest that conceptual inconsistencies may lurk in QM.  By insisting that Alice's certainties refer to the same underlying quantum state, we show, via our impossibility result, that the Frauchiger-Renner scenario cannot arise.

We still allow disagreement in QM, albeit of a milder nature than $0\operatorname{-}1$ disagreement.  So, we see our findings as challenging the notion that intersubjective disagreement should not be allowed in physics and, in this way, push in a direction broadly similar to the QBist and relational accounts of QM (Fuchs and Schack, 2013; Rovelli, 1996).  By introducing the tools of formal epistemics, our approach may suggest new treatments of these approaches.

In previous work, Contreras-Tejada et al.~(2021) formulate the agreement question in terms of a no-signaling box.  Within this environment, a quantum agreement theorem follows and fails only for superquantum boxes.  Liu, Chung, and Ramanathan (2024) generalize this result, while remaining in the commuting regime, to establish an agreement theorem for all quantum and almost-quantum regions.  The present paper allows the property projector to fail to commute with an agent's measurement.  This admits CCD, even when Alice's and Bob's measurements continue to commute with each other.  Leifer and Duarte (2026) treat the knowledge modality of the original classical Agreement Theorem, rather than the certainty modality, and establish the impossibility of common knowledge of disagreement in the setting of generalized probability theory (Barrett, 2007).  (We comment on the difference between these modalities in the Supplementary Material.)  Khrennikov and Basieva (2014) and Khrennikov (2015) study quantum-like agents interacting with quantum systems and are able to obtain CCD in this manner.

A natural extension of our framework would enrich the present protocol-agnostic CCD model by introducing a classical variable recording whether Alice performs her measurements before or after Bob, together with the agents' beliefs about this variable, beliefs about beliefs, and so on to higher levels as required.  For our CCD family, Alice's and Bob's measurements commute, so their order does not affect the distribution of their measurement outcomes.  However, the order affects the state Bob uses.  Conditional on his outcome $\beta_0$, Bob uses the state $\rho_t$ if Alice's measurement has not yet been incorporated.  If Alice has already measured (and Bob does not condition on her outcome), the relevant state is the dephased state of Equation \ref{eq:classical-mixture}.  Letting $\pi$ denote Bob's probability, given outcome $\beta_0$, that Alice has already measured, we obtain:
\begin{equation}
q_B = \frac{1 + (1 - \pi)t}{3}.
\end{equation}
Bob's assessment interpolates between the protocol-agnostic value $(1 + t)/3$ we calculated earlier and the post-measurement value $1/3$.  In a fuller epistemic model, Alice may be uncertain about Bob's value of $\pi$, Bob may be uncertain about Alice's beliefs about the value of $\pi$, and so on.  Summing up, the proposal is that the quantum model would supply the posterior associated with each measurement order, while a classical epistemic overlay would represent uncertainty over the order.  Another approach is to allow for indefinite causal order (Oreshkov, Costa, and Brukner, 2012; Chiribella et al., 2013; Cepollaro and Di Biagio, 2026).  The relationship between, on the one hand, our protocol-agnostic or interpolation approach, and, on the other hand, what is possible under indefinite order of measurement requires careful future examination.

A second open direction is to allow Alice and Bob to undertake measurements described by POVMs in place of the PVMs assumed in this paper.  Now, the post-measurement state can depend not only on the outcome but also on the particular quantum instrument realizing the POVM (G\"uhne et al., 2023).  Information about these instruments might or might not be accessible to the agents, and, in another extension, we would need to enrich the epistemics further to account for this feature.  A third direction is to identify general conditions under which the phenomenon of CCD does or does not arise.

In classical epistemics, the Agreement Theorem has been given a dynamic treatment in which Alice and Bob announce their conditional probabilities of the event of interest $E$, update on each other's announcements, and converge to common knowledge (Geanakoplos and Polemarchakis, 1982; Cave, 1983; Parikh and Krasucki, 1990; Krasucki, 1996).  Developing an analogous treatment in our quantum setting seems a priority.

We note that CCD, as a distinctive kind of quantum strangeness, also becomes, under the modern view, a new quantum resource.  We leave to future work possible applications to quantum computing and information (Nielsen and Chuang, 2011), to the design and performance of quantum-enabled economic games (Khan et al., 2018; Auletta et al., 2021), and to the detection of other manifestations of non-classicality (Smirne et al., 2019; Strasberg and Garc\'ia D\'iaz, 2019).  We end with what we think is an intriguing possibility: The value of CCD as a quantum resource may stem precisely from the fact that agents have common certainty of it.

\vspace{0.3in}

\noindent{\larger\larger\textbf{Supplementary Material}}

\renewcommand{\thesection}{A}
\setcounter{equation}{0}
\renewcommand\theequation{A.\arabic{equation}}
\setcounter{theorem}{0}
\renewcommand\thetheorem{A.\arabic{theorem}}
\setcounter{Ex}{0}
\renewcommand\theEx{A.\arabic{Ex}}
\setcounter{figure}{0}
\renewcommand\thefigure{A.\arabic{figure}}

\section*{Appendix A: Classical Agreement Theorem with the Certainty Modality}
We begin with a review of the classical Agreement Theorem.  The original statement due to Aumann (1976) is couched in terms of the knowledge modality and common knowledge.  Here, we establish a different classical baseline by using the certainty modality: Alice is certain of an event $E$ if she assigns probability $1$ to $E$ given her information (Brandenburger and Dekel, 1987; Monderer and Samet, 1989).  This shift is important in aligning our formal arguments with real-world laboratory procedures, since a probability-$1$ assignment to an event is a statement about local frequencies in a lab.

In addition, it can be shown that common knowledge of probability assessments yields agreement in any phase-space model with signed probabilities (Brandenburger et al., 2024, Theorem 4.1).  This result also aligns with Leifer and Duarte (2026).  By Abramsky and Brandenburger (2011, Theorem 5.9), the same is then true for any no-signaling model.  This gives another reason why we prefer the certainty modality.  Using knowledge erases the boundary between quantum and superquantum systems (cf.~the discussion of this boundary in the Introduction in the main text).

To proceed, fix a finite probability space $(\Omega, p)$ and information partitions ${\mathcal P}_A$ and ${\mathcal P}_B$ for Alice and Bob, respectively.  We fix versions of conditional probability $p\bigl(\cdot \, | \, {\mathcal P}_A(\omega)\bigr)$ and $p\bigl(\cdot \, | \, {\mathcal P}_B(\omega)\bigr)$ for Alice and Bob, respectively, where ${\mathcal P}_A(\omega)$ and ${\mathcal P}_B(\omega)$ are the cells of Alice's and Bob's partitions that contain $\omega$.  (As an aside, this formulation extends without change -- needing only the dominated convergence theorem -- to the general measurable case where $\sigma$-fields replace partitions.)

Now, given an event $E$ in $\Omega$ and probabilities $q_A$ and $q_B$, let:
\begin{align}
A_0 &= \bigl\{ \omega \in \Omega : p\bigl(E \, | \, {\mathcal P}_A(\omega)\bigr) = q_A \bigr\}, \label{eq:suppa}\\
B_0 &= \bigl\{ \omega \in \Omega : p\bigl(E \, | \, {\mathcal P}_B(\omega)\bigr) = q_B \bigr\}. \label{eq:suppb}
\end{align}
The event $A_0$ is the event that Alice assigns probability $q_A$ to $E$; likewise for the event $B_0$.  Next define recursively, for $n \ge 0$:
\begin{align}
A_{n+1} &= A_n \cap \bigl\{ \omega \in \Omega : p(B_n \, | \, {\mathcal P}_A(\omega)) = 1 \bigr\}, \label{eq:supp1} \\
B_{n+1} &= B_n \cap \bigl\{ \omega \in \Omega : p(A_n \, | \, {\mathcal P}_B(\omega)\bigr) = 1 \bigr\}, \label{eq:supp2}
\end{align}
and:
\begin{equation}
C_\infty = \bigcap_{n=0}^\infty A_n \cap \bigcap_{n=0}^\infty B_n.
\end{equation}
In words, it is common certainty at a state $\omega \in \Omega$ that Alice assigns probability $q_A$ to $E$ and Bob assigns probability $q_B$ to $E$ if $\omega \in C_\infty$.  Note that a special case is that $q_A = q_B  = 1$.  Then we can simply say that the event $E$ is common certainty at $\omega$.

\begin{theorem}
\textit{If it is common certainty at state $\omega$ that Alice assigns $q_A$ to $E$ and Bob assigns $q_B$ to $E$, and $p(C_\infty) > 0$, then $q_A = q_B$.}
\end{theorem}

\begin{proof}
Since $\Omega$ is finite, and the sets $A_n$ and $B_n$ are non-increasing, there is an integer $N$ such that $A_n = A_N$ and $B_n = B_N$ for all $n \ge N$.  It follows from Equations \ref{eq:supp1}--\ref{eq:supp2} that:
\begin{align}
A_N &\subseteq \bigl\{ \omega \in \Omega : p(B_N \, | \, {\mathcal P}_A(\omega)) = 1 \bigr\}, \label{eq:supp3}\\
B_N &\subseteq \bigl\{ \omega \in \Omega : p(A_N \, | \, {\mathcal P}_B(\omega)) = 1 \bigr\}. \label{eq:supp4}
\end{align}

Observe that $A_N$ is a union of cells of ${\mathcal P}_A$, which we write as $\cup\pi$.  Equation \ref{eq:supp3} then implies $p(B_N \, | \, \pi) = 1$ for each such $\pi$.  From this:
\begin{equation} \label{eq:supppi1}
p(A_N \cap B_N) = \sum_\pi p(\pi) \, p(B_N \, | \, \pi) = p(A_N).
\end{equation}
Next, we use $A_N \subseteq A_0$ to write $p(E \, | \, \pi) = q_A$ for the same partition cells $\pi$.  By the same argument as just given, we then get:
\begin{equation} \label{eq:supppi2}
p(E \cap A_N) = q_A \times p(A_N).
\end{equation}
Equation \ref{eq:supppi1} implies that for any event $E$:
\begin{equation}
p(E \cap A_N) = p(E \cap A_N \cap B_N),
\end{equation}
which, together with Equation \ref{eq:supppi2} yields:
\begin{equation}
q_A \times p(A_N) = p(E \cap A_N \cap B_N),
\end{equation}
or, using Equation \ref{eq:supppi1} again:
\begin{equation}
q_A \times p(A_N \cap B_N) = p(E \cap A_N \cap B_N).
\end{equation}

Running exactly the same argument with Bob in place of Alice, we obtain:
\begin{equation}
q_B \times p(A_N \cap B_N) = p(E \cap A_N \cap B_N).
\end{equation}
But $p(C_\infty) > 0$ implies $p(A_N \cap B_N) > 0$, so we conclude that $q_A = q_B$, as required.
\end{proof}

Figure \ref{fig:pool} shows that common certainty of probability assessments need not coincide with the shared assessment Alice and Bob would obtain after pooling their information.  Here, Alice's (resp.~Bob's) partition is the rows (resp.~columns).  The event of interest $E$ is $\{\omega_1, \omega_4\}$.  At any state, it is common certainty that Alice and Bob both assign probability $1/2$ to $E$.  Yet, if they pooled their information, their shared assessment would be $1$ or $0$, depending on the true state.

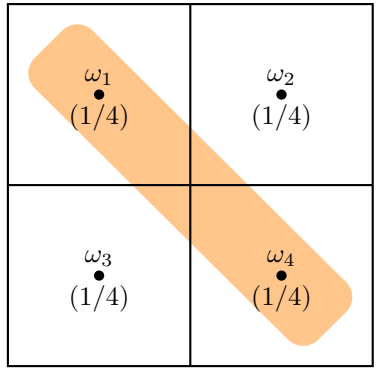
\begin{figure}[h]
\centering
\begin{tikzpicture}[scale=1.2]
\begin{scope}[shift={(2,2)}, rotate=-45]
\fill[orange!45, rounded corners=8pt] (-2.2,0.45) rectangle (2.2,-0.45);
\end{scope}
\draw[thick] (0,0) rectangle (4,4);
\draw[thick] (0,2) -- (4,2);
\draw[thick] (2,0) -- (2,4);
\fill (1,3) circle (1.6pt);
\node[above] at (1,3) {$\omega_1$};
\node[below] at (1,3) {$(1/4)$};
\fill (3,3) circle (1.6pt);
\node[above] at (3,3) {$\omega_2$};
\node[below] at (3,3) {$(1/4)$};
\fill (1,1) circle (1.6pt);
\node[above] at (1,1) {$\omega_3$};
\node[below] at (1,1) {$(1/4)$};
\fill (3,1) circle (1.6pt);
\node[above] at (3,1) {$\omega_4$};
\node[below] at (3,1) {$(1/4)$};
\end{tikzpicture}
\caption{Information pooling}
\label{fig:pool}
\end{figure}

\renewcommand{\thesection}{B}
\setcounter{equation}{0}
\renewcommand\theequation{B.\arabic{equation}}
\setcounter{theorem}{0}
\renewcommand\thetheorem{B.\arabic{theorem}}
\setcounter{Ex}{0}
\renewcommand\theEx{B.\arabic{Ex}}

\section*{Appendix B: Proof of the Commutative Agreement Theorem}

\begin{lemma} \label{lem:proj-stab}
For all $n$, the operators $A_n$ and $B_n$ are projectors.
\end{lemma}

\begin{proof}
The operators $A_0$ and $B_0$ are sums of mutually orthogonal
measurement projectors and, therefore, are again projectors.  Suppose that $A_n$ and $B_n$ are projectors.  The certainty operators $C_A(B_n)$ and $C_B(A_n)$ are also sums of mutually orthogonal measurement projectors and are projectors.  Moreover, $A_n$ and $C_A(B_n)$ both belong to the algebra generated by Alice's measurement projectors and therefore commute.  Likewise, $B_n$ commutes with $C_B(A_n)$.  It follows that $A_{n+1} = A_n \, C_A(B_n)$ and $B_{n+1} = B_n \, C_B(A_n)$ are products of commuting projectors and are again projectors.  This proves the result by induction.
\end{proof}

\begin{lemma} \label{lem:commute}
For all $n, m$, we have $[A_n , B_m] = 0$.  In particular, we have $[A_*, B_*] = 0$.
\end{lemma}

\begin{proof}
All $A_n$ lie in the algebra generated by Alice's measurement projectors and all $B_m$ lie in the algebra generated by Bob's measurement projectors.  Since each $P_A^i$ commutes with each $P_B^j$, the two algebras commute element-wise.  Therefore, every $A_n$ commutes with every $B_m$.
\end{proof}

The next lemma ensures that all normalization terms used below are strictly positive.

\begin{lemma} \label{lem:positive}
For all $n$, we have:
\begin{equation} \label{eq:lemma-b3}
\operatorname{Tr}(A_n \, \rho) \ge \operatorname{Tr}(A_* \, \rho) \ge \operatorname{Tr}(C_* \, \rho) > 0,
\end{equation}
and similarly for $B_n$. 
\end{lemma}

\begin{proof}
Since $A_{n+1} = A _n \, C_A(B_n)$ is a product of commuting projectors, we have $A_{n+1} \le A_n$ in the Loewner order.  Since the sequence stabilizes at $A_*$, we get $A_* \le A_n$ for every $n$.  By Lemma \ref{lem:commute}, $A_*$ and $B_*$ commute, so that $C_*= A_* \, B_*$ is the projector onto $\operatorname{ran}(A_*) \cap \operatorname{ran}(B_*)$, and therefore $C_* \le A_*$.  Combining these two inequalities gives $0 \le C_* \le A_* \le A_n$.

If $0 \le X \le Y$ and $\rho \ge 0$, then:
\begin{equation}
\operatorname{Tr}(Y \rho) - \operatorname{Tr}(X \rho)
= \operatorname{Tr}(\rho^{1/2} \, (Y - X) \, \rho^{1/2}) \ge 0.
\end{equation}
Applying this inequality to the preceding operator inequalities yields the first two parts of Inequality \ref{eq:lemma-b3}.  The final strict inequality holds by the definition of common certainty at state $\rho$.  This completes the proof for Alice, and the proof for Bob is analogous.
\end{proof}

\begin{lemma} \label{lem:Qsigma}
Let $Q$ be a projector and $\sigma$ a density operator. If
$\operatorname{Tr}(Q \, \sigma) = 1$, then:
\begin{equation}
Q \, \sigma \, Q = \sigma.
\end{equation}
\end{lemma}

\begin{proof}
The condition $\operatorname{Tr}(Q \, \sigma) = 1$ implies that $\sigma$ has support contained in $\operatorname{ran}(Q)$, from which it follows that $Q \, \sigma \, Q = \sigma$.
\end{proof}

Let $I$ index outcomes whose projectors occur in $A_*$:
\begin{equation}
A_* = \sum_{i \in I} P_A^i.
\end{equation}

\begin{lemma} \label{lem:invariance}
We have:
\begin{equation} \label{eq:lemma-b5}
\frac{B_* A_* \, \rho \, A_* B_*}{\operatorname{Tr}(B_* A_* \, \rho)} = \frac{A_* \, \rho \, A_*}{\operatorname{Tr}(A_* \, \rho)}.
\end{equation}
\end{lemma}

\begin{proof}
Since $A_*$ and $C_A(B_*)$ are sums of Alice's measurement projectors, the identity $A_* = A_* \, C_A(B_*)$ from the common-certainty recursion implies that every outcome projector occurring in $A_*$ also occurs in $C_A(B_*)$.  Therefore, we have $\Pr[B_*; P_A^i] = 1$ for every $i \in I$.  Also, since $A_* \leq A_0$, we have $\operatorname{Tr}(P_A^i \, \rho) > 0$ for every $i \in I$.  Define the density operator:
\begin{equation}
\sigma_* = \frac{A_* \, \rho \, A_*}{\operatorname{Tr}(A_* \, \rho)}.
\end{equation}
By the commutativity of $A_*$ and $B_*$, we obtain:
\begin{equation}
\operatorname{Tr}(B_* \, \sigma_*) = \frac{\operatorname{Tr}(B_* \, A_* \, \rho)}{\operatorname{Tr}(A_* \, \rho)} = \frac{\sum_{i \in I} \operatorname{Tr}(B_* \, P_A^i \, \rho \, P_A^i)}{\sum_{i \in  I} \operatorname{Tr}(P_A^i \, \rho)} =1.
\end{equation}
Lemma \ref{lem:Qsigma} then gives $B_* \, \sigma_* \, B_* = \sigma_*$.  The preceding calculation also shows that $\operatorname{Tr}(B_* \, A_* \, \rho) = \operatorname{Tr}(A_* \, \rho)$.  These two facts establish Equation \ref{eq:lemma-b5}, as required.
\end{proof}

\begin{lemma} \label{lem:b6}
We have:
\begin{equation}
\operatorname{Tr}\Bigl(P_E \frac{A_* \, \rho \, A_*}{\operatorname{Tr}(A_* \, \rho)}\Bigr) = \operatorname{Tr}\Bigl(P_E \frac{B_* A_* \, \rho \, A_* B_*}{\operatorname{Tr}(B_* A_* \, \rho)}\Bigr).
\end{equation}
\end{lemma}

\begin{proof}
Apply the linear functional $T \mapsto \operatorname{Tr}(P_E T)$ to both sides of Lemma \ref{lem:invariance}.
\end{proof}

\begin{lemma} \label{lem:b7}
For any $i \in I$ with $\operatorname{Tr}(P_A^i \, \rho) > 0$:
\begin{equation}
\operatorname{Pr}[P_E; P_A^i] = \operatorname{Tr}\Bigl(
P_E \frac{A_* \, \rho \, A_*}{\operatorname{Tr}(A_* \, \rho)} \Bigr).
\end{equation}
\end{lemma}

\begin{proof}
Since $A_* \le A_0 = Q_A(E; q_A)$, we have $\operatorname{Pr}[P_E; P_A^i] = q_A$ for every $i \in I$.  Since $P_E$ commutes with every $P_A^i$, the cross terms $\operatorname{Tr}(P_E \, P_A^i \, \rho \, P_A^k)$, with $k \not= i$, vanish by orthogonality, and we get:
\begin{equation}
\operatorname{Tr}(P_E \, A_* \, \rho \, A_*) = \sum_{i \in I} \operatorname{Tr}(P_E \, P_A^i \, \rho \, P_A^i) = \sum_{i \in I} q_A \times \operatorname{Tr}(P_A^i \, \rho) = q_A \times \operatorname{Tr}(A_* \, \rho).
\end{equation}
Dividing both sides by $\operatorname{Tr}(A_* \, \rho)$ yields the desired equality.
\end{proof}

The proof of the Commutative Agreement Theorem can now be completed, paralleling the main text.  By Lemmas \ref{lem:b6} and \ref{lem:b7},
\begin{equation}
q_A = \operatorname{Tr}\Bigl(P_E \frac{C_* \, \rho \, C_*}{\operatorname{Tr}(C_* \, \rho)}\Bigr),
\end{equation}
Applying the same argument with Alice and Bob interchanged gives the same expression for $q_B$, as required.

We note that, in the setting of Equations \ref{eq:com1}--\ref{eq:com2} of the main text, a proof of the Commutative Agreement Theorem could be obtained via the Gelfand representation theorem (Gelfand and Naimark, 1943).  This would involve some bookkeeping, and our operator proof works straightforwardly within the Hilbert-space formalism.  Our direct proof also sets the stage for the further results in the paper.

\renewcommand{\thesection}{C}
\setcounter{equation}{0}
\renewcommand\theequation{C.\arabic{equation}}
\setcounter{theorem}{0}
\renewcommand\thetheorem{C.\arabic{theorem}}
\setcounter{Ex}{0}
\renewcommand\theEx{C.\arabic{Ex}}
\renewcommand{\thefigure}{C\arabic{figure}}
\setcounter{figure}{0}

\section*{Appendix C: Additional Scenarios}

\begin{figure}[h] 
\centering \includegraphics[scale=0.6]{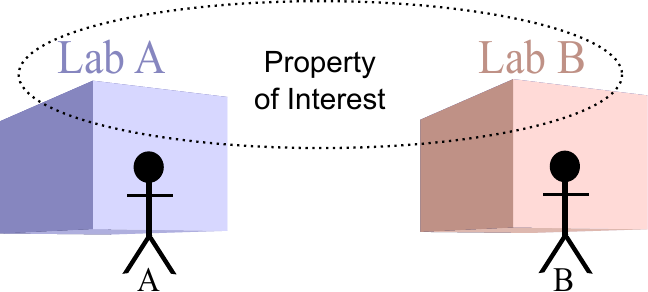}
\caption{Two-laboratory scenario}
\label{fig:2part}
\end{figure}

The three-lab setting in the main text provides a simple physical realization of the commutativity conditions used in the Commutative Agreement Theorem.  As a first variant, suppose that Alice and Bob have separate laboratories, with no third laboratory, and perform local measurements on a bipartite system.  The property $E$ may concern both subsystems and is represented by a projector on the joint Hilbert space.  See Figure \ref{fig:2part}.  Because the measurements are local, Alice's and Bob's projectors continue to commute with each other, and so Equation \ref{eq:com1} again holds.  The Agreement Theorem then holds in this setting whenever it makes sense to assume Equation \ref{eq:com2} outright.

\begin{figure}[h] 
  \centering \includegraphics[scale=0.6]{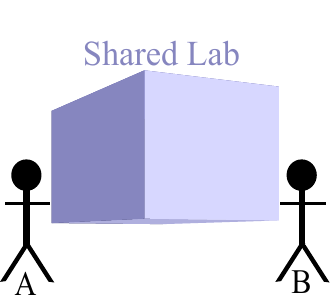}
  \caption{Single-laboratory scenario}
  \label{fig:1part}
\end{figure}

For the second variant, suppose that Alice and Bob perform measurements on a quantum system in a single shared laboratory.  See Figure \ref{fig:1part}.  In this setting, neither the commutativity of their measurements nor the commutativity of the measurements with $P_E$ automatically follows.  An Agreement Theorem follows if Equations \ref{eq:com1}--\ref{eq:com2} are assumed outright.

Finally, commutativity with $P_E$ is a sufficient but not necessary condition for agreement.  The same conclusion holds under the state-dependent conditions:
\begin{equation}
[P_A^i, \, \rho] = [P_B^j, \, \rho] = 0 \,\, \text{for all} \,\, i,j \,\, \text{and} \,\, [A_*, B_*] = 0.
\end{equation}

\renewcommand{\thesection}{D}
\setcounter{equation}{0}
\renewcommand\theequation{D.\arabic{equation}}
\setcounter{theorem}{0}
\renewcommand\thetheorem{D.\arabic{theorem}}
\setcounter{Ex}{0}
\renewcommand\theEx{D.\arabic{Ex}}
\setcounter{lemma}{0}
\renewcommand\thelemma{D.\arabic{lemma}}

\section*{Appendix D: Proof of the Non-Commutative Disagreement Theorem}
We begin with two preliminary lemmas establishing our set-up and then move to the proof of the theorem.

\begin{lemma}
For every $0 \le t \le 1$, $\rho_t$ is a density operator.
\end{lemma}

\begin{proof}
In the ordered basis $\{\ket{u}, \ket{v}\}$, the restriction
of $\rho_t$ to its support is:
\begin{equation}
\frac{1}{2}
\begin{pmatrix}
1 & t \\
t & 1
\end{pmatrix}.
\end{equation}
The matrix has eigenvalues $(1 + t)/2$ and $(1 - t)/2$, which are
nonnegative for $0 \le t \le 1$, and the trace is $1$, establishing the lemma.
\end{proof}

\begin{lemma}
Alice's and Bob's measurements commute with one another, Bob's measurement commutes with $P_E$, and Alice's measurement does not commute with $P_E$.
\end{lemma}

\begin{proof}
Alice's projectors act nontrivially only on $\mathcal H_A$, while
Bob's projectors act nontrivially only on $\mathcal H_B$, so their
measurements commute.  Since $P_E$ acts as the identity on
$\mathcal H_B$, it commutes with every $P_B^j$.  By contrast:
\begin{equation}
\bra{u} \, [P_A^0, P_E] \, \ket{v} = \bra{u} P_A^0 \, P_E \ket{v} - \bra{u} P_E \, P_A^0 \ket{v} = \frac{1}{3} \not= 0,
\end{equation}
so that Alice's measurement does not commute with $P_E$.
\end{proof}

\begin{lemma}
At $\rho_t$, Alice assigns probability $1/3$ to $E$
conditional on either of her positive-probability outcomes and Bob
assigns probability $(1 + t)/3$ to $E$ conditional on his unique
positive-probability outcome. That is:
\begin{align}
A_0 &= Q_A(E; 1/3) = P_A^0 + P_A^1, \\
B_0 &= Q_B(E; (1 + t)/3) = P_B^0.
\end{align}
\end{lemma}

\begin{proof}
Alice's outcome probabilities are:
\begin{equation}
\operatorname{Tr}(P_A^0 \, \rho_t) = \operatorname{Tr}(P_A^1 \, \rho_t) = \frac{1}{2}, \,\,\,\, \operatorname{Tr}(P_A^2 \, \rho_t) = 0.
\end{equation}
Conditional on outcomes $0$ and $1$, the updated states are $\ketbra{u}{u}$ and $\ketbra{v}{v}$, respectively.  Using:
\begin{equation}
\bra{u} P_E \ket{u} = \bra{v} P_E \ket{v} = \frac{1}{3},
\end{equation}
we find that Alice assigns probability $1/3$ to $E$, conditional on either positive-probability outcome.  Bob's outcome probabilities are:
\begin{equation}
\operatorname{Tr}(P_B^0 \, \rho_t) = 1, \,\,\,\, \operatorname{Tr}(P_B^1 \, \rho_t) = 0.
\end{equation}
Using $P_B^0 \, \rho_t \, P_B^0 = \rho_t$, we see that, conditional on outcome $\beta_0$, the updated state is again $\rho_t$.  We calculate:
\begin{equation}
\Pr[P_E; P_B^0] = \operatorname{Tr}(P_E \, \rho_t) = \frac{1}{2} \bra{u} P_E \ket{u} +\frac{1}{2} \bra{v} P_E \ket{v} + \frac{t}{2} \bra{u} P_E \ket{v} + \frac{t}{2} \bra{v} P_E \ket{u}.
\end{equation}
All four matrix elements are $1/3$, so that Bob assigns probability $(1 + t)/3$ to $E$ conditional on his positive-probability outcome.
\end{proof}

Note that the outcomes $P_A^2$ and $P_B^1$ are included because
$\Pi_A$ and $\Pi_B$ are complete projective measurements.  Also, the
subspace associated with $P_A^2$ supplies the third qutrit
direction appearing in $P_E$.  Since these outcomes have probability
$0$ under $\rho_t$, they do not enter $A_0$ or $B_0$.

\begin{lemma}
The common-certainty recursion stabilizes immediately:
\begin{equation} A_* = A_0, \,\,\,\, B_* = B_0, \,\,\,\, C_* = A_0 \, B_0, \,\, \text{and} \,\, \operatorname{Tr}(C_* \, \rho_t) = 1.
\end{equation}
\end{lemma}

\begin{proof}
We already noted that, conditional on Alice's positive-probability outcomes $P_A^0$ and $P_A^1$, the post-measurement states are $\ketbra{u}{u}$ and $\ketbra{v}{v}$, respectively.  Since $P_B^0 \ket{u} = \ket{u}$ and $P_B^0 \ket{v} = \ket{v}$, the support of each of these states is contained in $\operatorname{ran}(P_B^0)$.  Therefore:
\begin{equation}
\Pr[P_B^0; P_A^0] = \Pr[P_B^0; P_A^1] = 1.
\end{equation}
We also saw that, conditional on Bob's positive-probability outcome $P_B^0$, the post-measurement state is unchanged at $\rho_t$.  The support of $\rho_t$ is contained in $\operatorname{span}\{\ket{u}, \ket{v}\}$, which, since $(P_A^0 + P_A^1) \ket{u} = \ket{u}$ and $(P_A^0 + P_A^1) \ket{v} = \ket{v}$, is contained in $\operatorname{ran}(P_A^0 + P_A^1)$.  Therefore:
\begin{equation}
\Pr[P_A^0 + P_A^1; P_B^0] = 1.
\end{equation}
We have now shown that $C_A(P_B^0) = P_A^0 + P_A^1$ and $C_B(P_A^0 + P_A^1) = P_B^0$, which says that the common-certainty recursion stabilizes immediately.  Finally, because the support of $\rho_t$ is contained in $\operatorname{ran}((P_A^0 + P_A^1) \, P_B^0) = \operatorname{ran}(A_0 \, B_0)$, we have $\operatorname{Tr}(C_* \, \rho_t) = 1$.
\end{proof}

This completes the proof of the Non-Commutative Disagreement Theorem.

\renewcommand{\thesection}{E}
\setcounter{equation}{0}
\renewcommand\theequation{E.\arabic{equation}}
\setcounter{theorem}{0}
\renewcommand\thetheorem{E.\arabic{theorem}}
\setcounter{Ex}{0}
\renewcommand\theEx{E.\arabic{Ex}}
\setcounter{lemma}{0}
\renewcommand\thelemma{E.\arabic{lemma}}
\renewcommand{\thetable}{E\arabic{table}}
\setcounter{table}{0}
\renewcommand{\thefigure}{E\arabic{figure}}
\setcounter{figure}{0}

\section*{Appendix E: Proof of the $0\operatorname{-}1$ Impossibility Theorem}
We show that there is no outcome projector $P_A^i$ such that:
\begin{equation} \label{eq:three-conditions}
\operatorname{Tr}(P_A^i \, \rho) > 0, \,\, \Pr[P_E; P_A^i] = 1, \,\, \Pr[Q_B(E; 0); P_A^i] = 1.
\end{equation}
Suppose, toward a contradiction, that such an outcome exists.  Alice's updated state is: 
\begin{equation}
\sigma = \frac{P_A^i \, \rho P_A^i}{\operatorname{Tr}(P_A^i \, \rho)}.
\end{equation}
Write $S = \operatorname{supp}(\rho)$ and $K = \operatorname{supp}(\sigma)$.  Note that from:
\begin{equation}
P_A^i \, \rho \, P_A^i = \bigl(P_A^i \, \rho^{1/2}\bigr) \,
\bigl(P_A^i \rho^{1/2}\bigr)^\dagger,
\end{equation}
we get:
\begin{equation} \label{eq:ran}
K=\operatorname{ran}\bigl(P_A^i \, \rho^{1/2}\bigr).
\end{equation}

\begin{lemma} \label{lem:perp-0}
We have:
\begin{equation} \label{eq:perp}
K \cap S^\perp = \{0\}.
\end{equation}
\end{lemma}

\begin{proof}
Suppose $v \in K \cap S^\perp$.  By Equation \ref{eq:ran}, there is a vector $x \in \mathcal H$ such that $v = P_A^i \, \rho^{1/2} \, x$.  Setting $y = \rho^{1/2} \, x$, we get $y \in S$ and $v = P_A^i \, y$.  Since $P_A^i$ is an orthogonal projector:
\begin{equation}
\|v\|^2 = \langle P_A^i \, y, P_A^i \, y \rangle = \langle y, P_A^i \, y \rangle = \langle y, v\rangle = 0,
\end{equation}
where the last equality follows from $y \in S$ and $v \in S^\perp$.
Therefore, $v = 0$, which establishes the result.
\end{proof}

The second condition in Equation \ref{eq:three-conditions} gives $\operatorname{Tr}(P_E \, \sigma) = 1$.  It follows from Lemma \ref{lem:Qsigma} that the support of $\sigma$ is contained in the range of $P_E$.  That is:
\begin{equation} \label{eq:K-ran-1}
K \subseteq \operatorname{ran}(P_E).
\end{equation}
Likewise, the third condition in Equation \ref{eq:three-conditions} gives $\operatorname{Tr}(Q_B(E; 0) \, \sigma) = 1$, from which:
\begin{equation} \label{eq:K-ran-2}
K \subseteq \operatorname{ran}(Q_B(E; 0)).
\end{equation}

We now use the definition of Bob's probability-$0$ projector.  Let:
\begin{equation}
J_0 = \bigl\{ j : \operatorname{Tr}(P_B^j \, \rho) > 0 \,\, \text{and} \,\, \text{Pr}[P_E; P_B^j] = 0\bigr\},
\end{equation}
so that:
\begin{equation} \label{eq:sum-here}
Q_B(E; 0) = \sum_{j \in J_0} P_B^j.
\end{equation}

\begin{lemma} \label{lem:perp}
For every $j \in J_0$ and $v \in K$, we have:
\begin{equation}
P_B^j \, v \in S^\perp.
\end{equation}
\end{lemma}

\begin{proof}
Fix $j \in J_0$.  Bob's probability-$0$ assessment gives:
\begin{equation}
\operatorname{Tr}(P_E \, P_B^j \, \rho \, P_B^j) = 0.
\end{equation}
Using $P_E^2 = P_E$ and cyclicity of the trace, we obtain:
\begin{equation}
0 = \operatorname{Tr}(P_E \, P_B^j \, \rho \, P_B^j \, P_E) =
\operatorname{Tr}\bigl[(\rho^{1/2} \, P_B^j \, P_E)^\dagger
(\rho^{1/2} \, P_B^j \, P_E)\big].
\end{equation}
The operator inside the trace is positive semidefinite, and so the trace vanishes only if the operator itself vanishes, which says:
\begin{equation} \label{eq:vanish}
\rho^{1/2} \, P_B^j \, P_E = 0.
\end{equation}
Fix $v \in K$.  Equation \ref{eq:K-ran-1} gives $P_Ev = v$, so that, using Equation \ref{eq:vanish}, we get:
\begin{equation}
\rho^{1/2} \, P_B^j \, v = \rho^{1/2 \, }P_B^j \, P_Ev = 0.
\end{equation}
Since $\operatorname{ker}(\rho^{1/2}) = S^\perp$, we conclude $P_B^j \, v \in S^\perp$, as required.
\end{proof}

The proof of the $0\operatorname{-}1$ Impossibility Theorem can now be completed. Fix $v \in K$.  Summing over $j \in J_0$ and using Equation \ref{eq:sum-here} and Lemma \ref{lem:perp}, we get:
\begin{equation} \label{eq:perp-more}
Q_B(E; 0) \, v \in S^\perp.
\end{equation}
Equation \ref{eq:K-ran-2} gives $Q_B(E; 0) \, v = v$.  It follows that $v \in S^\perp$.  Since $v \in K$ is arbitrary, this yields $K \subseteq S^\perp$.  Lemma \ref{lem:perp-0} now implies $K = \{0\}$, contradicting the fact that $\sigma$ is a density operator and so has nonzero support.  We conclude that there is no outcome $i$ for which all three conditions of Equation \ref{eq:three-conditions} hold.  It remains to connect this conclusion to the operator formulation of the theorem.  Let:
\begin{align}
I_E &= \big\{i : \operatorname{Tr}(P_A^i \, \rho) > 0 \,\, \text{and} \,\, \Pr[P_E; P_A^i] = 1\bigr\}, \\
I_B &= \big\{i : \operatorname{Tr}(P_A^i \, \rho) > 0 \,\, \text{and} \,\, \Pr[Q_B(E; 0); P_A^i]=1\bigr\}.
\end{align}
By mutual orthogonality of Alice's outcome projectors:
\begin{equation}
Q_A(E; 1) \, C_A(Q_B(E; 0)) = \sum_{i \in I_E \cap I_B} P_A^i.
\end{equation}
But the preceding argument shows that $I_E \cap I_B = \emptyset$, so that:
\begin{equation}
Q_A(E; 1) \, C_A(Q_B(E; 0)) = 0,
\end{equation}
and:
\begin{equation}
\operatorname{Tr}\bigl(Q_A(E; 1) \, C_A(Q_B(E; 0)) \, \rho\bigr) = 0,
\end{equation}
as required.

We next present a no-signaling box (Popescu and Rohrlich, 1994) that exhibits 0-1 disagreement under a natural way of handling epistemics in this setting.  (The example builds on Contreras-Tejada et al., 2021.)
\vspace{0.1in}

\noindent\textbf{Example E.1} \textit{In Table \ref{tab:nosignaling-box}, there are three possible measurements $a, b, e$, each with outcome $0$ or $1$.  The compatible measurements are the pairs $(a, b)$, $(a, e)$, and $(b, e)$.  The probabilities of each of the sets of four joint outcomes are given in each row.  The no-signaling condition is satisfied.  For example, if the measurement pair $(a, b)$ is undertaken, the marginal probability that $b$ yields $1$ is $1/4 + 1/4 = 1/2$.  If the measurement pair $(b, e)$ is undertaken, this marginal is $1/2 + 0 = 1/2$.  The other no-signaling equalities can be checked similarly.}

\begin{table}[h]
\centering
\renewcommand{\arraystretch}{1.4}
\begin{tabular}{c|cccc}
& $(0,0)$ & $(1,0)$ & $(0,1)$ & $(1,1)$ \\
\hline
$(a,b)$ & $\tfrac{1}{2}$ & $0$ & $\tfrac{1}{4}$ & $\tfrac{1}{4}$ \\
$(a,e)$ & $\tfrac{1}{2}$ & $0$ & $\tfrac{1}{4}$ & $\tfrac{1}{4}$ \\
$(b,e)$ & $0$ & $\tfrac{1}{2}$ & $\tfrac{1}{2}$ & $0$ \\
\end{tabular}
\caption{No-signaling box exhibiting 0\,-1 disagreement}
\label{tab:nosignaling-box}
\end{table}

\textit{Now, let the event of interest be $e = 1$ and suppose Alice obtains the outcome $a = 1$.  Then, from the second row:}
\begin{equation}
\operatorname{Pr}[e = 1 | a = 1] = \frac{1/4}{0 + 1/4} = 1.
\end{equation}
\textit{From the third row:}
\begin{equation}
\operatorname{Pr}[e = 1 | b = 0] = \frac{1/2}{0 + 1/2} = 1, \,\,\,\, \operatorname{Pr}[e = 0 | b = 1] = \frac{1/2}{1/2 + 0} = 1.
\end{equation}
\textit{Finally, from the first row:}
\begin{equation}
\operatorname{Pr}[b = 1 | a = 1] = \frac{1/4}{0 + 1/4} = 1.
\end{equation}
\textit{We conclude that if Alice measures $a$ and gets the outcome $1$, she is certain that $e = 1$.  She is also certain that $b = 1$.  From this, she is certain that Bob is certain that $e = 0$.  We obtain $0\operatorname{-}1$ disagreement.}

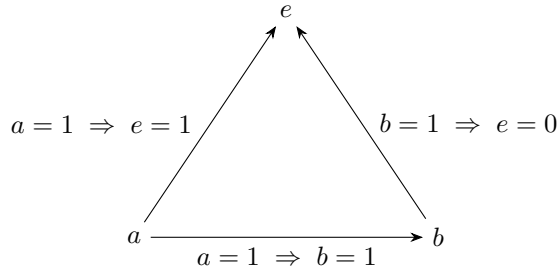
\begin{figure}[h]
\centering
\begin{tikzpicture}[>=Stealth, node distance=2cm]
  \node (A) at (0,0) {$a$};
  \node (B) at (4,0) {$b$};
  \node (E) at (2,3) {$e$};
  
  \draw[->] (A) -- node[below] {$a=1 \;\Rightarrow\; b=1$} (B);
  \draw[->] (B) -- node[right] {$\,\, b=1 \;\Rightarrow\; e=0$} (E);
  \draw[->] (A) -- node[left]  {$a=1 \;\Rightarrow\; e=1 \,\,$} (E);

\end{tikzpicture}
\caption{Implication structure}
\label{fig:specker}
\end{figure}

Interestingly, as Figure \ref{fig:specker} indicates, the structure of this no-signaling example is exactly that of the Specker triangle (Specker, 1960).  This connection between contextuality and $0\operatorname{-}1$ disagreement would seem deserving of further study.

We next recast this example in quasi-classical terms using signed probabilities on phase space (following Abramsky and Brandenburger, 2011, Theorem 5.4, and Brandenburger et al., 2024).  Though an equivalent setting, the epistemic analysis may be more transparent this way.  Let $X =\{a, b, e\}$ be the set of elementary measurements and $O =\{0, 1\}$ the set of outcomes.  The phase space is the set of global assignments $\Omega = O^X$ and the signed probability measure $\lambda$ on $\Omega$ is depicted in Table \ref{tab:signed} below.

\begin{table}[h!]
\centering
\renewcommand{\arraystretch}{1.4}
\begin{tabular}{c|ccc|c}
\hline
$\omega_i$ & $a$ & $b$ & $e$ & $\lambda$ \\
\hline
$\omega_0$ & 0 & 0 & 0 & $+\tfrac{3}{16}$ \\
$\omega_1$ & 0 & 0 & 1 & $+\tfrac{5}{16}$ \\
$\omega_2$ & 0 & 1 & 0 & $+\tfrac{5}{16}$ \\
$\omega_3$ & 0 & 1 & 1 & $-\tfrac{1}{16}$ \\
$\omega_4$ & 1 & 0 & 0 & $-\tfrac{3}{16}$ \\
$\omega_5$ & 1 & 0 & 1 & $+\tfrac{3}{16}$ \\
$\omega_6$ & 1 & 1 & 0 & $+\tfrac{3}{16}$ \\
$\omega_7$ & 1 & 1 & 1 & $+\tfrac{1}{16}$ \\
\hline
\end{tabular}
\caption{Signed phase-space probability measure}
\label{tab:signed}
\end{table}

It can be checked that this table realizes the no-signaling box of Table \ref{tab:nosignaling-box}.  Also, we give Alice the partition $\{ \{\omega : a(\omega) = 0\}, \{\omega : a(\omega) = 1\}\}$ and Bob the partition $\{ \{\omega : b(\omega) = 0\}, \{\omega : b(\omega) = 1\}\}$.  The event of interest is $E = \{\omega : e(\omega) = 1\}$.  At state $\omega_7$, Alice assigns conditional probability $1$ to $E$.  The event that Bob assigns conditional probability $0$ to $E$ is $\{\omega : b(\omega) = 1\}$.  At $\omega_7$, Alice assigns conditional probability $1$ to $\{\omega : b(\omega) = 1\}$.  This establishes $0\operatorname{-}1$ disagreement in the phase-space setting.

\renewcommand{\thesection}{F}
\setcounter{equation}{0}
\renewcommand\theequation{F.\arabic{equation}}
\setcounter{theorem}{0}
\renewcommand\thetheorem{F.\arabic{theorem}}
\setcounter{Ex}{0}
\renewcommand\theEx{F.\arabic{Ex}}
\setcounter{lemma}{0}
\renewcommand\thelemma{F.\arabic{lemma}}

\section*{Appendix F: Proofs of Bounds on Disagreement}
To prove the quantum Inequality \ref{eq:bound-a} in the main text, we build on the classical argument in Hellman (2013).  To begin, assume $c_A = \operatorname{Tr}(C_* \, \sigma) > 0$ and $c_B = \operatorname{Tr}(C_* \, \tau) > 0$.  Also, write $e_A = \operatorname{Tr}(P_E C_* \, \sigma)$ and $e_B = \operatorname{Tr}(P_E C_* \, \tau)$, so that $q_A = e_A/c_A$, $q_B = e_B/c_B$, and:
\begin{equation}
|q_A - q_B| = \frac{|c_B e_A - c_A e_B|}{c_A c_B}.
\end{equation}

Next, set $x = c_Be_A$, $y = c_Ae_B$, and $z = c_Be_B$, and use the triangle inequality:
\begin{equation}
|x - y| \le |x - z| + |z - y|,
\end{equation}
to obtain:
\begin{equation} \label{eq:holder}
|q_A - q_B| \le  \frac{|c_B e_A - c_B e_B|}{c_A c_B} + \frac{|c_B e_B - c_A e_B|}{c_A c_B}.
\end{equation}
Rewrite the first numerator in Equation \ref{eq:holder} as:
\begin{equation}
|c_B e_A - c_B e_B| = c_B \,|e_A - e_B| = c_B \, \bigl|\operatorname{Tr}(P_E C^*(\sigma - \tau))\bigr|.
\end{equation}
Now use the operator analog to the classical H\"older's inequality:
\begin{equation}
\bigl|\operatorname{Tr}(X^\dagger Y)\bigr| \le \|X\|_1 \, \|Y\|_\infty,
\end{equation}
and the fact that $\|P_E C_*\|_\infty \le 1$ to write:
\begin{equation}
|e_A - e_B| \le \|\sigma - \tau\|_1,
\end{equation}
from which:
\begin{equation} \label{eq:num1}
|c_B e_A - c_B e_B| \le c_B \, \|\sigma - \tau\|_1.
\end{equation}

Turning to the second numerator in Equation \ref{eq:holder}, write:
\begin{equation}
|c_B e_B - c_A e_B| = e_B \, |c_B - c_A| = e_B \, \bigl|\operatorname{Tr}(C^*(\sigma - \tau))\bigr|.
\end{equation}
A second application of H\"older's inequality together with $\|C_*\|_\infty = 1$ yields:
\begin{equation}
|c_B - c_A| \le \|\sigma - \tau\|_1.
\end{equation}
Also, from $P_E \le \mathbb{I}$ we obtain:
\begin{equation}
e_B = \operatorname{Tr}(P_E C_* \, \tau) \le \operatorname{Tr}(C_* \, \tau) = c_B,
\end{equation}
so that:
\begin{equation} \label{eq:num2}
|c_B e_B - c_A e_B| \le c_B \, \|\sigma - \tau\|_1.
\end{equation}

Substituting Inequalities \ref{eq:num1} and \ref{eq:num2} into Inequality \ref{eq:holder}, we find:
\begin{equation} \label{eq:final1}
|q_A - q_B| \le \frac{c_B \, \|\sigma - \tau\|_1 + c_B \, \|\sigma - \tau\|_1}{c_A c_B} = \frac{2\, \|\sigma - \tau\|_1 }{c_A}.
\end{equation}
Rerunning the argument with Alice and Bob interchanged gives:
\begin{equation} \label{eq:final2}
|q_A - q_B| \le  \frac{2\, \|\sigma - \tau\|_1 }{c_B}.
\end{equation}
Putting Inequalities \ref{eq:final1}--\ref{eq:final2} together yields Inequality \ref{eq:bound-a}.

To prove Inequality \ref{eq:bound-b} in the main text, we adapt the classical argument in Monderer and Samet (1989).  Using Inequality \ref{eq:born-epsilon}, we can retrace the steps in Section S.2 to obtain:
\begin{equation} \label{eq:epsilon}
\operatorname{Tr}(B_* A_* \, \rho) \ge (1 - \epsilon) \operatorname{Tr}(A_* \, \rho).
\end{equation}
Now, decompose $A_*$ as $A_* = C_* + A_*(\mathbb{I} - B_*)$, from which:
\begin{equation} \label{eq:total}
\operatorname{Tr}(P_E A_* \, \rho) = \operatorname{Tr}(P_E C_* \, \rho) + \operatorname{Tr}\bigl(P_E A_*(\mathbb{I} - B_*) \, \rho\bigr).
\end{equation}
Letting:
\begin{equation}
x = \operatorname{Pr}[P_E; C_*], \,\,\,\, y = \operatorname{Pr}[P_E; A_*(\mathbb{I} - B_*)], \,\,\,\, \alpha = \operatorname{Pr}[C_*; A_*],
\end{equation}
we can use Equation \ref{eq:total} to write:
\begin{equation} \label{eq:exact}
q_A = x \alpha + y (1 - \alpha).
\end{equation}
Inequality \ref{eq:epsilon} tells us that $\alpha \ge 1 - \epsilon$, from which Equation \ref{eq:exact} yields the estimate:
\begin{equation}
q_A \ge x(1 - \epsilon) + 0,
\end{equation}
while, at the same time, we can also estimate:
\begin{equation}
q_A \le x\cdot1 + 1\cdot\epsilon.
\end{equation}
Since $x$ is symmetric in Alice and Bob, we can repeat these bounds for $q_B$.  It follows that:
\begin{equation}
|q_A - q_B | \le x + \epsilon - x(1 - \epsilon),
\end{equation}
from which, since $x \le 1$, Inequality \ref{eq:bound-b} follows.

\end{document}